\documentclass[1p]{elsarticle}
\usepackage{amsmath}
\usepackage{amssymb}
\usepackage{amsthm}
\usepackage{color}
\usepackage{graphicx}
\usepackage{xspace}
\usepackage{tikz}
\usepackage[noend]{algpseudocode}
\usepackage{algorithm}

\newcommand{\firefighter}{{\sc{Firefighter}}\xspace}

\newcommand{\onesat}{{\sc{Cubic Monotone 1-In-3-Sat}}\xspace}
\newcommand{\ecover}{{\sc{Exact Cover by 3-Sets}}\xspace}

\newcommand{\psaved}{\ensuremath{k}\xspace}
\newcommand{\pburned}{\ensuremath{k}\xspace}

\newlength{\atextwidth}
\newcommand{\p}{\ensuremath{\mathsf{P}}}
\newcommand{\np}{\ensuremath{\mathsf{NP}}}

\newcommand{\wone}{\ensuremath{\mathsf{W[1]}}}

\newtheorem{lemma}{Lemma}
\newtheorem{theorem}{Theorem}
\newtheorem{corollary}{Corollary}
\newtheorem{remark}{Remark}

\newcommand{\ie}{\textit{i.e.}~}

\DeclareMathOperator{\pw}{pw}
\DeclareMathOperator{\cw}{cw}

\DeclareMathOperator{\bw}{bw}
\DeclareMathOperator*{\argmax}{arg\,max}
\DeclareMathOperator*{\argmin}{arg\,min}

\newcommand{\problemdec}[3]{
  \vspace{1mm}
\noindent\fbox{
  \begin{minipage}{0.95\linewidth}
  \begin{tabular*}{\linewidth}{@{\extracolsep{\fill}}lr} #1 \\ \end{tabular*}
  {\bf{Input:}} #2  \\
  {\bf{Question:}} #3
  \end{minipage}
  }
  \vspace{1mm}
}

\usetikzlibrary{arrows,decorations.pathmorphing,backgrounds,positioning,fit}
\usetikzlibrary{shapes}
\usetikzlibrary{arrows,calc}

\tikzset{
  treenode/.style = {align=center, inner sep=0pt, text centered,font=\sffamily},
  arn_n/.style = {circle, draw=black, fill=black,inner sep=0pt,minimum size=4pt},
  arn_r/.style = {circle, draw=black, fill=black},
  arn_x/.style = {circle, draw=black, fill=black},
}

\usepackage[pdfdisplaydoctitle,menucolor=orange!40!black,filecolor=magenta!40!black,urlcolor=blue!40!black,linkcolor=red!40!black,citecolor=green!40!black,colorlinks]{hyperref}

\begin{document}





\begin{frontmatter}

\title{The Firefighter Problem: A Structural Analysis} 

\author[janka]{Janka Chleb\'ikov\'a}
\ead{janka.chlebikova@port.ac.uk}

\author[morgan]{Morgan Chopin\fnref{fn1}}
\ead{chopin@lamsade.dauphine.fr}

\address[janka]{University of Portsmouth, School of Computing, 1.24, Buckingham Building, Lion Terrace, Portsmouth PO1 3HE, United Kingdom}
\address[morgan]{Universit\'e Paris-Dauphine, LAMSADE, Place du Mar\'echal de Lattre de Tassigny,\\ 75775 Paris Cedex~16, France}
\fntext[fn1]{A major part of this work done during a three-month visit of the University of Portsmouth supported by the ERASMUS program.}

\begin{abstract}
\sloppy
We consider the complexity of the firefighter problem where~${b \geq 1}$ firefighters are available at each time step.
This problem is proved \np-complete even on trees of degree at most three and budget one~\cite{finbow2007} and on trees of bounded degree~$b+3$ for any fixed budget~$b \geq 2$~\cite{bazgan2012}. 

In this paper, we provide further insight into the complexity landscape of the problem by showing that the pathwidth and the maximum degree of the input graph govern its complexity.
More precisely, we first prove that the problem is \np-complete even on trees of pathwidth at most three for any fixed budget~$b \geq 1$. We then show that the problem turns out to be fixed parameter-tractable with respect to the combined parameter ``pathwidth'' and ``maximum degree'' of the input graph. 
\end{abstract}


\end{frontmatter}








\section{Introduction}

The firefighter problem was introduced by Hartnell~\cite{hartnell1995} and received considerable attention in a series of papers~\cite{anshelevich2012,cai2008,develin2007,finbow2007,hartnell2000,IKM11,king2010,mac03,ng2008,Costa2013}.
In its original version, a fire breaks out at some
vertex of a given graph. At each time step, one
vertex can be protected by a firefighter and then the fire
spreads to all unprotected neighbors of the vertices on fire. 
The process ends when the fire can no longer spread. At the end all vertices that are not on fire are considered as saved. The
objective is at each time step to choose a vertex which
will be protected by a firefighter such that a maximum number of
vertices in the graph is saved at the end of the process. 
In this paper, we consider a more general version which allows us to protect $b \geq 1$ vertices at each step (the value $b$ is called \textit{budget}). 

The original firefighter problem was proved to be \np-hard for bipartite graphs
\cite{mac03}, cubic graphs \cite{king2010} and unit disk graphs \cite{fomin2012}. 
Finbow~et~al.~\cite{finbow2007} showed that the problem is \np-hard even on trees. More precisely, they proved the following dichotomy theorem:
 the problem is
\np-hard even for trees of maximum degree three and it is solvable
in polynomial-time for graphs with maximum degree three, provided
that the fire breaks out at a vertex of degree at most two. Furthermore, the problem is polynomial-time solvable for caterpillars and so-called P-trees \cite{mac03}.
Later, Bazgan~et~al.~\cite{bazgan2012} extended the previous results by showing that the general firefighter problem is \np-hard even for trees of maximum degree~$b+3$ for any fixed budget~$b \geq 2$ and polynomial-time solvable on $k$-caterpillars.
From the approximation point of view, the problem is $\frac{e}{e-1}$-approximable on trees ($\frac{e}{e-1} \approx 1.5819$) \cite{cai2008} and
it is not $n^{1-\varepsilon}$-approximable on general graphs for any~${\varepsilon > 0}$ unless $\p = \np$ \cite{anshelevich2012}. Moreover for trees in which each non-leaf vertex has at most three children, the firefighter problem is $1.3997$-approximable \cite{IKM11}. Very recently, Costa~et~al.~\cite{Costa2013} extended the  $\frac{e}{e-1}$-approximation algorithm on trees to the case where the fire breaks out at $f>1$ vertices and $b>1$ firefighters are available at each step.
From a parameterized perspective, the problem is \wone-hard with respect to the natural parameters ``number of saved vertices'' and ``number of burned vertices''~\cite{ouripec,ourisaac}.
Cai~et~al.~\cite{cai2008} gave first fixed-parameter tractable algorithms and polynomial-size kernels for trees for each of the following parameters: ``number of saved vertices'', ``number of saved leaves'', ``number of burned vertices'', and ``number of protected vertices''.

In this paper we show that the complexity of the problem is governed by the maximum degree and the pathwidth of the input graph. In \autoref{cont:sec:prel}, we first provide the formal definition of the problem as well as some preliminaries. In \autoref{sec:pw-hard}, we complete the hardness picture of the problem on trees by proving that it is also \np-complete on trees of pathwidth three. We note that the given proof is also a simpler proof of the \np-completeness of the problem on trees. In \autoref{sec:algo}, we devise a parameterized algorithm with respect to the combined parameter ``pathwidth'' and ``maximum degree'' of the input graph. The conclusion is given in \autoref{sec:conclusion}.

\section{Preliminaries} \label{cont:sec:prel}








\sloppy
\paragraph{Graph terminology} Let $G=(V,E)$ be an \textit{undirected graph} of order $n$.  For a subset $S\subseteq V$, $G[S]$ is the induced subgraph of $G$. 
The \textit{neighborhood} of a vertex $v \in V$, denoted by $N(v)$, is the set of all neighbors of $v$. 
We denote by $N^{k}(v)$ the set of vertices which are at distance at most $k$ from $v$.
The \textit{degree} of a vertex $v$ is denoted by $\deg_G(v)$ and the \emph{maximum degree} of the graph~$G$ is denoted by~$\Delta(G)$.

A \textit{linear layout} of $G$ is a bijection $\pi : V \to \{1,\ldots,n\}$. For convenience, we express~$\pi$ by the list $L=(v_1,\ldots,v_n)$ where $\pi(v_i) = i$. Given a linear layout $L$,  we denote the distance between two vertices in $L$ by $d_L(v_i, v_j) = |i-j|$.

The \textit{cutwidth} $\cw(G)$ of $G$ is the minimum ${k \in \mathbb{N}}$ such that the vertices of $G$ can be arranged in a linear layout $L=(v_1,\ldots,v_n)$ in such a way that, for every ${i \in \{1,\ldots,n-1\}}$, there are at most $k$ edges between $\{v_1,\ldots,v_i\}$ and $\{v_{i+1},\ldots,v_n\}$.

The \textit{bandwidth} $\bw(G)$ of $G$ is the minimum ${k\in \mathbb{N}}$ such that the vertices of $G$ can be arranged in a linear layout $L=(v_1,\ldots,v_n)$ so that $d_L(v_i, v_j) \leq k$ for every edge $v_iv_j$ of $G$.

A path decomposition $\mathcal{P}$ of $G$ is a pair $(P, \mathcal{H})$  where $P$ is a path with node set $X$ and $\mathcal{H} = \{H_x : {x\in X}\}$ is a family of subsets of $V$ such that the following conditions are met
\begin{enumerate}
\item $\bigcup_{{x \in X}} H_x = V.$
\item For each $uv \in E$ there is an $x \in X$ with $u, v \in H_x$.
\item For each~${v \in V}$, the set of nodes~$\{{x : {x \in X}  \mbox{ and } {v \in H_x} }\}$ induces a subpath of~$P$.
\end{enumerate}
The width of a path decomposition $\mathcal{P}$ is $\max_{x \in X} |H_x | - 1$.
The pathwidth~$\pw(G)$ of a graph~$G$ is the minimum width over all possible path decompositions of~$G$.

We may skip the argument of $\pw(G)$, $\cw(G)$, $\bw(G)$ and $\Delta(G)$ if the graph~$G$ is clear from the context.

A \textit{star} is a tree consisting of one vertex, called the \textit{center} of the star, adjacent to all the other vertices.

\paragraph{Parameterized complexity} 
\sloppy
Here we only give the basic notions on parameterized complexity used in this paper, for more background the reader is referred to \cite{DF99,Nie06}. The parameterized complexity is a framework which provides a new way to express the computational complexity of problems. A decision problem parameterized by a problem-specific parameter~$k$ is called \emph{fixed-parameter tractable} if there exists an algorithm that solves it in time~$f(k) \cdot n^{O(1)}$ where~$n$ is the instance size. The function~$f$ is typically super-polynomial and only depends on~$k$. In other words, the combinatorial explosion is confine into~$f$. We may sometimes say that a problem is fixed-parameter tractable with respect the \emph{combined parameter} $k_1$,$k_2$, \ldots, and $k_p$  meaning that the problem can be solved in time~$f(k_1, k_2, \ldots, k_p) \cdot n^{O(1)}$.

\paragraph{Problems definition} 
We start with an informal explanation of the propagation process for the firefighter problem.
Let~$G=(V,E)$ be a graph of order~$n$ with a vertex~$s \in V$, let $b \in \mathbb{N}$ be a \textit{budget}.
At step $t = 0$, a fire breaks out at vertex~$s$ and~$s$ starts burning.
At any subsequent step~$t > 0$ the following two phases are performed in sequence:
\begin{enumerate}
\item \emph{Protection phase} : The firefighter protects at most $b$ vertices not yet on fire.
\item \emph{Spreading phase} : Every unprotected vertex which is adjacent to a burned vertex starts burning.
\end{enumerate}
Burned and protected vertices remain burned and protected until the propagation process stops, respectively. 
The propagation process stops when in a next step no new vertex can be burned.
We call a vertex saved if it is either protected or if all paths from any burned vertex to it contains at least one protected vertex.
Notice that, until the propagation process stops, there is at least one new burned vertex at each step. This leads to the following obvious lemma.

\begin{lemma} \label{lem:bounded-steps}
The number of steps before the propagation process stops is less or equal to the total number of burned vertices.
\end{lemma}

A \textit{protection strategy} (or simply \textit{strategy})~$\Phi$ indicates which vertices to protect at each step until the propagation process stops.
Since there can be at most $n$ burned vertices, it follows from \autoref{lem:bounded-steps} that the propagation unfolds in at most $n$ steps. 
We are now in position to give the formal definition of the investigated problem.


\problemdec{The \firefighter problem:}{A graph $G=(V,E)$, a vertex $s\in V$, and positive integers $b$ and~$\psaved$.}{Is there a strategy for an instance $(G,s,b,k)$ with respect to budget~$b$ such that at most \psaved vertices are burned if a fire breaks out at~$s$?}




When dealing with trees, we use the following observation which is a straightforward adaptation of the one by MacGillivray and Wang for the case~$b > 1$~\cite[Section~4.1]{mac03}.

\begin{lemma} \label{lem:trees:base}
Among the strategies that maximizes the number of saved vertices (or equivalently minimizes the number of burned vertices) for a tree, there exists one that protects vertices adjacent to a burned vertex at each time step.
\end{lemma}

Throughout the paper, we assume all graphs to be connected since otherwise we can simply consider the component where the initial burned vertex~$s$ belongs to.

\section{Firefighting on path-like graphs} \label{sec:pw-hard}

Finbow~et~al.~\cite{finbow2007} showed that the problem is \np-complete even on trees of degree at most three. However, the constructed tree in the proof has an unbounded pathwidth. 
In this section, we show that the \firefighter problem is \np-complete even on trees of pathwidth three.
For that purpose we use the following problem.

\problemdec{The \onesat problem:}{A CNF formula in which every clause contains exactly and only three positive literals and every variable appears in exactly three clauses.}{Is there a satisfying assigment (a truth assignment such that each clause has exactly one true literal) for the formula?}

The \np-completeness of the above problem is due to its equivalence with the \np-complete \ecover problem~\cite{GJ79}.

\begin{theorem} \label{th:pw}
The \firefighter problem is \np-complete even on trees of pathwidth three and budget one.
\end{theorem}

\begin{proof}
Clearly, \firefighter belongs to \np. Now we provide a polynomial-time reduction from \onesat.

In the proof, a guard-vertex is a star with $\pburned$ leaves where the center is adjacent to a vertex of a graph. It is clear that if at most $\pburned$ vertices can be burned then the guard-vertex has to be saved.

\sloppy Let~$\phi$ be a formula of~\onesat with $n$ variables~$\{x_1, \ldots, x_n\}$ and $m$ initial clauses~$\{c_1, \ldots, c_m\}$.
Notice that a simple calculation shows that~$n=m$. First, we extend~$\phi$ into a new formula~$\phi'$ by adding~$m$ new clauses as follows. For each clause~$c_j$ we add the clause~$\bar{c}_j$ by taking negation of each variable of~$c_j$.
A satisfying assignment for~$\phi'$ is then a truth assignment such that each clause~$c_j$ has exactly one true literal and each clause~$\bar{c}_j$ has exactly two true literals.
It is easy to see that~$\phi$ has a satisfying assignment if and only if~$\phi'$ has one.

Now we construct an instance~$I' = (T,s,1,\pburned)$ of \firefighter from~$\phi'$ as follows (see \autoref{fig:pw}).
We start with the construction of the tree $T$, the value of~$\pburned$ will be specified later.

\begin{itemize}
\item Start with a vertex set $\{s=u_1,u_2,\ldots,u_p\}$ and edges of~$\{su_2,u_2u_3,\ldots,u_{p-1}u_{p}\}$ where~$p = 2n-1$ and add two degree-one vertices~$v_{x_i}$ and~$v_{\bar{x}_i}$ adjacent to~$u_{2i-1}$ for every~$i \in \{1,\ldots,n\}$.
\end{itemize}



\noindent
Then for each~$i \in \{1,\ldots,n\}$ in two steps:

\begin{itemize}
\item Add a guard-vertex~$g_i$ (resp.~$\bar{g}_i$) adjacent to~$v_{x_i}$ (resp.~$v_{\bar{x}_i}$).

\item At each vertex~$v_{x_i}$ (resp.~$v_{\bar{x}_i}$) root a path of length~$2 \cdot (n-i)$ at~$v_{x_i}$ (resp.~$v_{\bar{x}_i}$) in which the endpoint is adjacent to three degree-one vertices (called literal-vertices) denoted by~$\ell^{x_i}_1$,~$\ell^{x_i}_2$, and~$\ell^{x_i}_3$ (resp.~$\ell^{\bar{x}_i}_1$,~$\ell^{\bar{x}_i}_2$, and~$\ell^{\bar{x}_i}_3)$. Each literal-vertex corresponds to an occurence of the variable~$x_i$ in an initial clause of~$\phi$. Analogously, the literal-vertices~$\ell^{\bar{x}_i}_1$,~$\ell^{\bar{x}_i}_2$, and~$\ell^{\bar{x}_i}_3$ represent the negative literal~$\bar{x}_i$ that appears in the new clauses of~$\phi'$.
\end{itemize}

\noindent
Notice that each leaf of the constructed tree so far is at distance exactly~$p+1$ from~$s$.

\begin{itemize}
\item For each variable~$x_i$ (resp.~$\bar{x}_i$),~$i \in \{1,\ldots,n\}$,
there are exactly three clauses containing $x_i$ (resp. $\bar{x}_i$). 
Let~$c_j$ (resp.~$\bar{c}_j$),~$j \in \{1,\ldots,m\}$, be the first one of them.
Then root a path~$Q^{x_i}_j$ (resp.~$Q^{\bar{x}_i}_j$) of length~$3 \cdot (j-1)$ at~$\ell^{x_i}_1$ (resp.~$\ell^{\bar{x}_i}_1$), and add a guard-vertex~$g^{x_i}_j$ adjacent to the endpoint of~$Q^{x_i}_j$.
To the endpoint of~$Q^{\bar{x}_i}_j$ (i) add a degree-one vertex~$d^{\bar{x}_i}$ (a dummy-vertex) and (ii) root a path of length~$3$ where the last vertex of the path is a guard vertex~$g^{\bar{x}_i}_j$
Repeat the same for two other clauses with $x_i$ (resp.~$\bar{x}_i$) and~$\ell^{x_i}_2$,~$\ell^{x_i}_3$ (resp.~$\ell^{\bar{x}_i}_2$,~$\ell^{\bar{x}_i}_3$).
\end{itemize}
To finish the construction, set~$\pburned = p + \frac{n}{2} (11n+7)$.



In what follows, we use~\autoref{lem:trees:base} and thus we only consider strategies that protect a vertex adjacent to a burned vertex at each time step. Recall that the budget is set to one in the instance $I'$.
Now we show that there is a satisfying assignment for~$\phi'$ if and only if there exists a strategy for~$I'$ such that at most~$\pburned$ vertices in~$T$ are burned.

``$\Rightarrow$'' : Suppose that there is a satisfying assignment~$\tau$ for~$\phi'$. We define the following strategy~$\Phi_{\tau}$ from~$\tau$.
At each step~$t$ from $1$ to $p+1$, if~$t$ is odd then protect~$v_{\bar{x}_{\lceil t/2 \rceil}}$ if~$x_{\lceil t/2 \rceil}$ is true otherwise protect~$v_{x_{\lceil t/2 \rceil}}$.
If~$t$ is even then protect the guard-vertex~$g_{\lceil t/2 \rceil}$ if~$v_{\bar{x}_{\lceil t/2 \rceil}}$ has been protected, otherwise protect~$\bar{g}_{\lceil t/2 \rceil}$.
At the end of time step~$p+1$, the number of burned vertices is exactly~$p + \sum_{i=1}^n (3+2(n-i)+1) = p + 3n + n^2$. 
Moreover, the literal-vertices that are burned in~$T$ correspond to the true literals in~$\phi'$.
Thus, by construction and since~$\tau$ statisfies~$\phi'$, the vertices adjacent to a burning vertex are exactly one guard-vertex~$g^{x_a}_1$, two dummy vertices~$d^{\bar{x}_b}, d^{\bar{x}_c}$ and~$3n-1$ other vertices where~$x_a \vee x_b \vee x_c$ is the first clause,~$a, b, c \in \{1,\ldots,n\}$. At step~$p+2$, we must protect the guard vertex~$g^{x_a}_1$. During the steps~$p+3$ and~$p+4$, the strategy must protect one vertex lying on the path~$D^{\bar{x}_b}_1$ and~$D^{\bar{x}_c}_1$, respectively. Thus~$3(3n-3) + 5 = 9n-4$ more vertices are burned at the end of step~$p+4$. 
More generally, from time step $p+3(j-1)+2$ to $p+3(j-1)+4$, for some~$j \in\{1,\ldots,m\}$, the strategy~$\Phi_{\tau}$ must protect a guard-vertex~$g^{x_a}_{j}$ and one vertex of each path~$D^{\bar{x}_b}_{j}$ and~$D^{\bar{x}_c}_{j}$, where~$x_a, x_b, x_c$ appear in the clause~$c_j$,~$a, b, c \in \{1,\ldots,n\}$. Thus~$9(n- (j - 1)) -4$ vertices get burned.
It follows that the number of burned vertices from step~$p+2$ to~$p+3m+1$ is~$\sum_{j=1}^{m}[ 9(n- (j - 1)) -4 ] = \frac{9}{2}m(m+1) - 4m$.
Putting all together, we arrive at a total of~$p + 3n + n^2 + \frac{9}{2}m(m+1) - 4m = p + \frac{n}{2} (11n+7) = k$ burned vertices.

``$\Leftarrow$'': 
Conversely, assume that there is no satisfying assignment for~$\phi'$.
Observe first that any strategy~$\Phi$ for~$I'$ protects either~$v_{x_i}$ or~$v_{\bar{x}_i}$ for each~$i \in \{1,\ldots,n\}$. 
As a contradiction, suppose that there exists~$i \in \{1,\ldots, n\}$ such that~$\Phi$ does not protect~$v_{x_i}$ and~$v_{\bar{x}_i}$. Then in some time step both~$v_{x_i}$ and~$v_{\bar{x}_i}$ get burned. Hence, it is not possible to protect both~$g_i$ and~$\bar{g}_i$ and at least one will burn implying that more than~$\pburned$ vertices would burn, a contradiction. Furthermore,~$v_{x_i}$ and~$v_{\bar{x}_i}$ cannot be both protected otherwise we would have protected a vertex not adjacent to a burned vertex at some step.
Now consider the situation at the end of step~$p+1$.
By the previous observation, the literal-vertices that are burned in~$T$ can be interpreted as being the literals in~$\phi'$ set to true.
As previously, the number of burned vertices so far is exactly~$p + \sum_{i=1}^n (3+2(n-i)+1) = p + 3n + n^2$.
Let~$n_g$ and~$n_d$ be the number of guard-vertices and dummy-vertices adjacent to a burned vertex, respectively.
As it follows from the previous construction, we know that~$n_g = 3 - n_d$ with~$0 \leq n_g \leq 3$ and~$0 \leq n_d \leq 3$. We have the following possible cases:
\begin{itemize}
\item[(1)] $n_g > 1$. In this case, a guard-vertex gets burned and hence more than~$\pburned$ vertices would burn.
\item[(2)] $n_g = 1$. Let~$g^{x_a}_1$ be that guard-vertex and let~$d^{\bar{x}_b}, d^{\bar{x}_c}$ be the~$n_d = 3 - n_g = 2$ dummy-vertices where~$x_a, x_b, x_c$ are variables of the first clause. At time step~$i = p+2$, we must protect~$g^{x_a}_1$. Furthermore, during the step $i = p+3$ (resp. $i = p+4$), any strategy must protect a vertex lying on the path~$D^{\bar{x}_b}_1$ (resp. $D^{\bar{x}_c}_1$). Indeed, if a strategy does otherwise then at least one guard-vertex~$g^{\bar{x}_b}_1$ or $g^{\bar{x}_c}_1$ gets burned. Thus~$2$ dummy-vertices are burned.
\item[(3)] $n_g = 0$. Hence we have exactly~$n_d = 3 - n_g = 3$ dummy-vertices~$d^{\bar{x}_a}, d^{\bar{x}_b}, d^{\bar{x}_c}$ adjacent to burned vertices.
Using a similar argument as before, we know that during the step $i = p+2$ (resp. $i = p+3$, $i = p+4$), a strategy must protect a vertex lying on the path~$D^{\bar{x}_a}_1$ (resp. $D^{\bar{x}_b}_1$, $D^{\bar{x}_c}_1$). Thus~$3$ dummy-vertices are burned.
\end{itemize}
Notice that at step~$p+5$, we end up with a similar situation as in step~$p+2$.
Now consider an assignment for~$\phi'$. Since~$\phi'$ is not satisfiable, therefore~$\phi$ is not satisfiable as well. There are two possibilities:
\begin{itemize}
\item There exists a clause~$c_j$ in~$\phi$ with more than one true literal. Thus, we end up with case~(1) and there is no strategy for~$I'$ such that at most~$k$ vertices are burned.
\item There is a clause~$c_j$ in~$\phi$ with only false literals. This corresponds to the case~(3) and the number of burned vertices would be at least~$1+p + \frac{n}{2} (11n+7)$ (at least one extra dummy-vertex gets burned) giving us a total of at least~$\pburned+1$ burned vertices. Hence there is no strategy for~$I'$ where at most~$k$ vertices are burned.
\end{itemize}

It remains to prove that the pathwidth of~$T$ is at most three. 
To see this, observe that any subtree rooted at~$v_{x_i}$ or~$v_{\bar{x}_i}$ has pathwidth two. 
Let~$P_{x_i}$ and~$P_{\bar{x}_i}$ be the paths of the path-decompositions of these subtrees, respectively.
We construct the path-decomposition for~$T$ as follows.
For every~$i \in \{1,\ldots,n-1\}$, define the node~$B_{i} = \{u_{2i-1},u_{2i},u_{2i+1}\}$.
Extend all nodes of the paths~$P_{x_i}$ and~$P_{\bar{x}_i}$ to ~$P'_{x_i}$ and~$P'_{\bar{x}_i}$ by adding the vertex~$u_{2i-1}$ inside it. Finally, connect the paths~$P'_{x_1}$,~$P'_{\bar{x}_1}$ and the node~$B_{1}$ to form a path and continue in this way with~$P'_{x_2}$,~$P'_{\bar{x}_2}$,~$B_{2}$,~$P'_{x_3}$,~$P'_{\bar{x}_3}$,~$B_{3}$,~\ldots,~$B_{n-1}$,~$P'_{x_n}$,~$P'_{\bar{x}_n}$.
This completes the proof.
\end{proof}

\begin{figure*}[!t]
\centering
\includegraphics[scale=0.8]{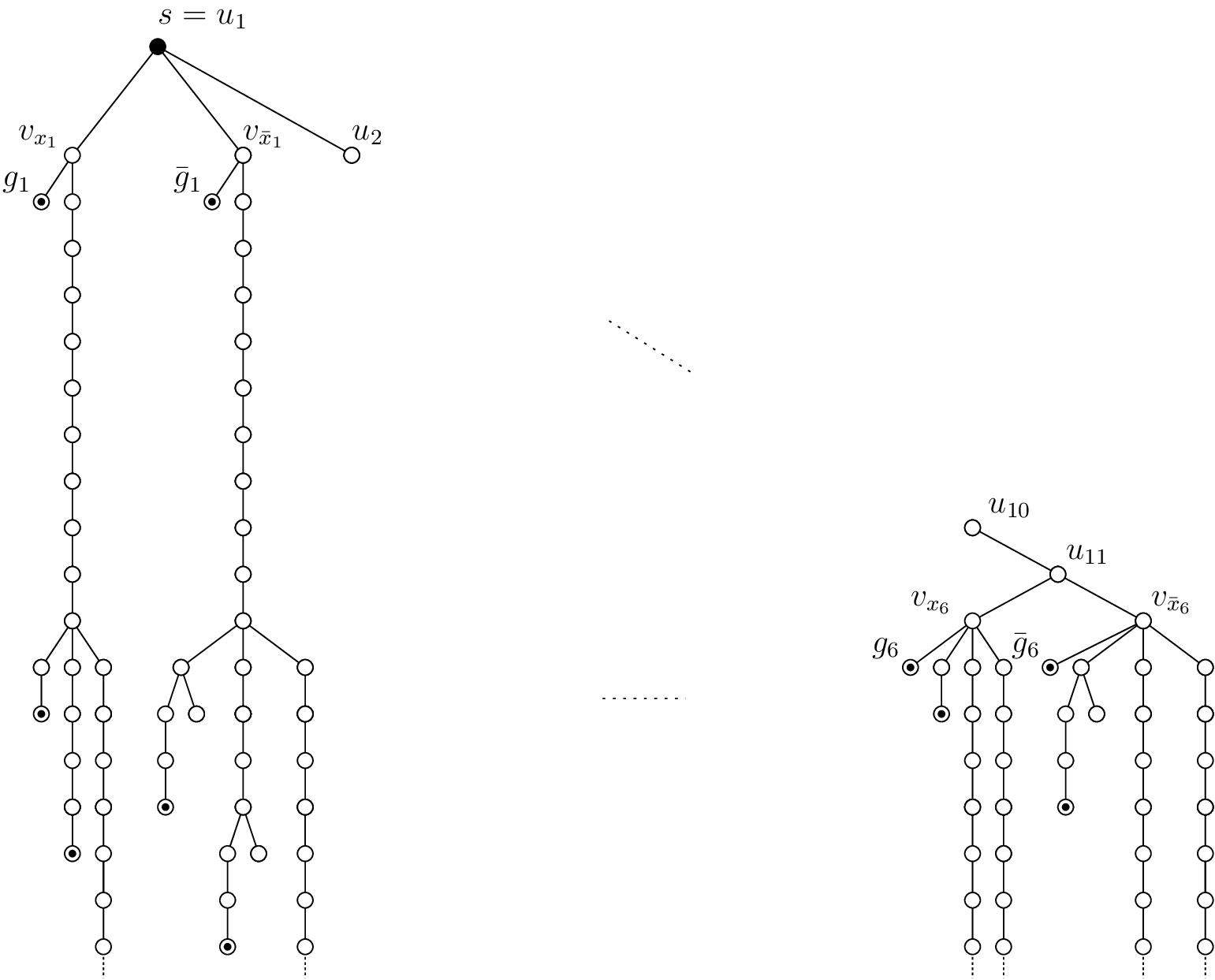}
\caption{An example of part of a tree constructed from the formula $\phi = (x_1 \vee x_3 \vee x_6) \wedge (x_1 \vee x_2 \vee x_3) \wedge (x_3 \vee x_4 \vee x_5)  \wedge (x_2 \vee x_4 \vee x_5) \wedge (x_1 \vee x_4 \vee x_6)\wedge (x_2 \vee x_6 \vee x_5)$. Guard vertices are represented by a dot within a circle.}
\label{fig:pw}
\end{figure*}

We can generalize the previous result to any fixed budget $b \geq 1$ as follows.

\begin{corollary} \label{cor:npc-pw-anybudget}
For any fixed budget $b \geq 1$, the \firefighter problem is \np-complete even on trees of pathwidth three.
\end{corollary}

\begin{proof}
We start from the reduction of~\autoref{th:pw} and alter the tree $T$ as follows. Let~$w_1$ be the vertex~$s$ (corresponds also to~$u_1$). Add a path $\{w_1w_2,w_2w_3,\ldots,w_{5n}w_{5n+1}\}$ to $T$ together with~$b-1$ guard-vertices added to each~$w_i$. First, one can easily check that the pathwidth remains unchanged since the added component has pathwidth two and is only connected to the root~$s$. Second, it can be seen that at each time step, only one firefighter can be placed ``freely'' as the other $b-1$ firefighters must protect $b-1$ guard-vertices. It follows that we end up to a similar proof as for~\autoref{th:pw}. This completes the proof.
\end{proof}









\section{Path-like graphs of bounded degree}  \label{sec:algo}

As previously shown, for any fixed budget $b \geq 1$, the \firefighter problem is \np-complete on trees of bounded degree~$b+3$~\cite{finbow2007,bazgan2012} and on trees of bounded pathwidth three (\autoref{th:pw}). It is thus natural to ask for the complexity of the problem when both the degree and the pathwidth of the input graph are bounded. In what follows, we answer this question positively. A first step toward this goal is to use the following combinatorial characterization of the number of burned vertices in a graph.






\begin{theorem} \label{lem:boundburn}
Consider a graph of pathwidth $\pw$ and maximum degree $\Delta$. 

\sloppy
\noindent
If the number of initially burned vertices is bounded by~$f_1(\pw, \Delta)$ for some function~$f_1$ then 
there exists a protection strategy such that at most~${f_2(\pw, \Delta) \geq f_1(\pw, \Delta)}$ vertices are burned for some function $f_2$.
\end{theorem}

\begin{proof}
First we  prove the following claim: Consider a graph of cutwidth $\cw$. If the number of initially burned vertices is bounded by~$g_1(\cw)$ for some function~$g_1$ then there exists a protection strategy such that at most~$g_2(\cw) \geq g_1(\cw)$ vertices are burned for some function~$g_2$. We will prove this by induction on~$\cw$.

The claim is obviously true when the cutwidth is $0$ since the graph cannot contain any edge.
Suppose now that the claim is true for any graph of cutwidth at most~$k$, $k > 0$. 
We show that it also holds for a graph of cutwidth~$k+1$.
Let~$H=(V,E)$ be such a graph and~$F \subseteq V$ be the set of initially burned vertices with $|F| \leq g_1(\cw(H))$ for some function~$g_1$.
Consider a linear layout~$L =  (v_1, \ldots, v_n)$ of~$H$ such that for every~$i = 1,\ldots,n-1$, there are at most~$k+1$ edges between~$\{v_1,\ldots,v_i\}$ and~$\{v_{i+1},\ldots,v_n\}$. For every~$s \in F$ and~ $i \geq 0$, we define inductively the following sets, where $R_0(s) = L_0(s) = \{s\}$
\begin{equation} 
R_i(s) = \left\{
    \begin{array}{ll}
        \{s=v_k,v_{k+1},\ldots,v_{k'}\} & \mbox{if } \exists v_{k'} \in N^{i}(s) \mbox{ : } v_{k'} = \displaystyle\argmax_{v \in N^{i}(s)} d_{L}(s,v) \\
         R_{i-1}(s) & \mbox{otherwise}
    \end{array}
\right.
\end{equation}
\begin{equation} 
L_i(s) = \left\{
    \begin{array}{ll}
        \{s=v_k,v_{k-1},\ldots,v_{k'}\} & \mbox{if } \exists v_{k'} \in N^{i}(s) \mbox{ : } v_{k'} =\displaystyle\argmin_{v \in N^{i}(s)} d_{L}(s,v) \\
         L_{i-1}(s) & \mbox{otherwise}
    \end{array}
\right.
\end{equation}
We are now in position to define the set~$B_i(s)$, called a \textit{bubble}, by~$B_i(s) = L_i(s) \cup R_i(s)$ for all~$i \geq 0$. Informally speaking, the bubble~$B_i(s)$ corresponds to the \textit{effect zone} of~$s$ after~$i$ steps of propagation \ie every burned vertex inside the bubble is due to the vertex~$s$. The idea of the proof is to show that every bubble can be ``isolated'' from the rest of the graph in a bounded number of steps by surrounding it with firefighters (see \autoref{fig:bubbles}). We then show that the inductive hypothesis can be applied on each bubble which will prove the theorem.

Let $s_1, s_2 \in F$.
We say that two bubbles~$B_i(s_1)$ and~$B_j(s_2)$ for some $i,j \geq 0$ \textit{overlap} if there exists an edge~$uv \in E$ with~$u \in B_i(s_1)$ and~$v \in B_j(s_2)$. In this case, we can \textit{merge} two bubbles into one \ie we create a new bubble which is the union of~$B_i(s_1)$ and~$B_j(s_2)$.

Let us consider an initially burned vertex~$s \in F$ and its bubble~$B_{2 \cdot \cw(H)}(s)$.
First, merge~$B_{2 \cdot \cw(H)}(s)$ with every other bubble~$B_{2 \cdot \cw(H)}(s')$ with~$s' \in F$ that possibly overlap into a new one~$B'_{2 \cdot \cw(H)}(s)$. By definition, we know that the number of edges with an endpoint in~$B'_{2 \cdot \cw(H)}(s)$ and the other one in~$V \setminus B'_{2 \cdot \cw(H)}(s)$ is less or equal to~$2 \cdot \cw(H)$. 
Thus, we define the strategy that consists in protecting one vertex~$v \in V \setminus B'_{2 \cdot \cw(H)}(s)$ adjacent to a vertex in~$B'_{2 \cdot \cw(H)}(s)$ at each step~$t=1, \ldots, 2 \cdot \cw(H)$.
Let~$F'$ be the set of vertices burned at step~$2 \cdot \cw(H)$.
Since~$\Delta(H) \leq 2 \cdot \cw(H)$, we deduce that~$|F'|$ is less or equal to~$|F|\cdot\Delta(H)^{2 \cdot \cw(H)} \leq g_1(\cw(H))\cdot(2 \cdot \cw(H))^{2 \cdot \cw(H)}$ hence bounded by a function of~$\cw(H)$. Let us consider the subgraph~$H' = H[B'_{2 \cdot \cw(H)}(s)]$.
Observe that we can safely remove every edge~$uv$ from~$H'$ for which~$u,v \in F'$. Indeed, such edge cannot have any influence during the subsequent steps of propagation.
By the definition of a bubble, this implies that the cutwidth of~$H'$ is decreased by one and thus is now at most~$k$. 
Therefore, we can apply our inductive hypothesis to~$H'$ which tells us that there is a strategy for~$H'$ such that at most~$g'_2(\cw(H'))$ vertices are burned for some function~$g'_2$. By \autoref{lem:bounded-steps}, this strategy uses at most~$g'_2(\cw(H'))$ steps to be applied. It follows that the number of burned vertices in~$H$ after applying this strategy is at most the number of burned vertices from step~$1$ to the step~$2 \cdot \cw(H) + g'_2(\cw(H'))$ which is~$|F|\cdot\Delta(H)^{2 \cdot \cw(H) + g'_2(\cw(H'))} \leq g_1(\cw(H))\cdot(2 \cdot \cw(H))^{2 \cdot \cw(H) + g'_2(\cw(H'))}$ which is bounded by a function of~$\cw(H)$. From now on, one can see that the previous argument can be applied iteratively to each bubble. Since the number of bubbles is bounded by~$g_1(\cw(H))$ (there is at most one bubble for each vertex initially on fire), we deduce that the total number of burned vertices is bounded by~$g_2(\cw(H))$ some function~$g_2$. This concludes the proof of the claim.

We are now in position to prove the theorem. Let~$G$ be a graph. Suppose that the number of initially burned vertices in~$G$ is at most~$f_1(\pw(G), \Delta(G))$ for some function~$f_1$. We know that~$\pw(G) \leq \cw(G)$ and~$\Delta(G) \leq 2 \cdot \cw(G)$~\cite{korach93}. Thus the number of burned vertices is at most~$f_1'(\cw(G))$ for some function $f_1'$. From the above claim we deduce that there exists a strategy such that at most~$f_2'(\cw(G))$ vertices get burned. Since~$\cw(G) \leq \pw(G) \cdot \Delta(G)$~\cite{chung89}, it follows that the number of burned vertices is bounded by~$f_2(\pw(G), \Delta(G))$ for some function $f_2$.  
This completes the proof.
\end{proof}

\begin{figure*}[t]
\centering
\includegraphics{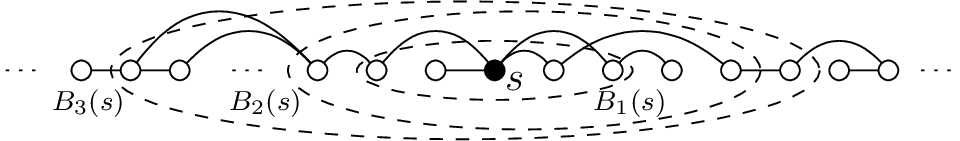}
\caption{A linear layout of a graph of cutwidth two. Dashed ellipses represent the bubbles associated to an initially burned vertex $s$.}
\label{fig:bubbles}
\end{figure*}

\begin{remark}
Notice that~\autoref{lem:boundburn} is still valid even if the number of firefighters available at each step is not the same (for example if there are $b_1$ firefighters at time step one, $b_2$ firefighters during the second time step, etc.).
\end{remark}

In~\cite{ourisaac} the authors proved that the \firefighter problem is fixed-parameter tractable with respect to the combined parameter~$k$ and the budget~$b$. Let $(G,s,b,\psaved)$ be an instance of \firefighter where~$G$ has maximum degree~$\Delta$.
We can derive the following algorithm: If $b \geq \Delta$ then protect all the vertices in~$N(s)$ at time step one; otherwise, apply the algorithm from~\cite{ourisaac}.
We then easily obtain the following

\begin{theorem} \label{th:burnkanddeg}
The \firefighter problem is fixed-parameter tractable with respect to the combined parameter~$k$ and~``maximum degree'' of the input graph.
\end{theorem}





We are now in position to give the main result of this section.

\begin{theorem} \label{th:fpt-pw-degree}
The \firefighter problem is fixed-parameter tractable with respect to the combined parameter ``pathwidth'' and ``maximum degree'' of the input graph.
\end{theorem}

\begin{proof}
Let $(G,s,b,\psaved)$ be an instance of \firefighter where~$G$ has maximum degree~$\Delta$ and pathwidth~$\pw$. 
We design the following algorithm. For each value $k'=1,\ldots,\psaved$ run the $f(k',\Delta) \cdot n^{O(1)}$-time algorithm of~\autoref{th:burnkanddeg}: If the algorithm returns ``yes'' then return ``yes''. If the algorithm has returned the answer ``no'' for all $k'=1,\ldots,\psaved$ then return ``no''. 

Using~\autoref{lem:boundburn}, we know that there exists a function $f'$ such that if~$k' \geq f'(\pw, \Delta)$ the algorithm will necessarily returns ``yes'' and stops. It follows that the algorithm is called at most $f'(\pw, \Delta)$ times. The overall running time is then bounded by
\begin{eqnarray*}
O(f'(\pw, \Delta) \cdot f(k',\Delta) \cdot n^{O(1)}) & = & O(f'(\pw, \Delta) \cdot f(f'(\pw, \Delta),\Delta) \cdot n^{O(1)})\\
  & = & f''(\pw, \Delta) \cdot n^{O(1)}
\end{eqnarray*}
for some function $f''$. This completes the proof.
\end{proof}


From the proof of~\autoref{lem:boundburn} and the fact that $\cw(G) \leq \frac{\bw(G) (\bw(G) + 1)}{2}$~\cite{bodlaender88} for any graph~$G$, we easily deduce the following corollary.

\begin{corollary}
The \firefighter problem is fixed-parameter tractable with respect to the parameters  ``cutwidth'' and ``bandwidth''.
\end{corollary}

\section{Conclusion} \label{sec:conclusion}

In this paper we showed that the \firefighter problem is \np-complete even on trees of pathwidth three but fixed-parameter tractable with respect to the combined parameter ``pathwidth'' and ``maximum degree'' of the input graph. The combination of these two results with the \np-completeness of the problem on trees of bounded degree~\cite{finbow2007} indicates that the complexity of the problem depends heavily on the degree and the pathwidth of the graph. We left as an open question whether the problem is polynomial-time solvable on graphs of pathwidth two. 

\bibliographystyle{abbrv}
\bibliography{main}

\end{document}